\documentclass[runningheads]{llncs}
\usepackage{bbding}
\usepackage{graphicx}
\usepackage{hyperref}
\usepackage{amsmath,amssymb,amsfonts}
\usepackage{algorithm}
\usepackage{listings}
\usepackage{comment}
\usepackage{bbding}
\usepackage{tikz}
\usepackage{subfig}
\usepackage{wrapfig}
\usepackage{rotating}
\usepackage{multirow}
\usepackage{listings}
\usepackage{placeins}
\usepackage{xspace}
\usepackage{siunitx}

\usepackage{enumitem}
\newcommand*\mycirc[1]{%
	\begin{tikzpicture}[baseline=(C.base)]
		\node[draw,circle,inner sep=1pt,minimum size=3ex](C) {#1};
\end{tikzpicture}}
\usepackage{booktabs}
\newcommand{\bt}{\fontseries{b}\selectfont}
\usetikzlibrary{matrix}
\usetikzlibrary{positioning, fit}
\usetikzlibrary{arrows,shapes,patterns,calc,automata}
\tikzstyle{trans}=[font=\footnotesize]
\tikzstyle{single_s}=[rectangle,draw, double, trans, minimum size=8mm]
\tikzstyle{single_s2}=[rectangle,draw, double, trans, minimum
width =12mm, minimum height=6mm]
\tikzstyle{single_is}=[rectangle,draw, double, fill=yellow, trans, minimum
size=8mm]
\tikzstyle{kwc_s}=[rectangle,draw, trans, minimum
width =18mm, minimum height=8mm]
\tikzstyle{kwc_is}=[rectangle,draw, trans, fill=yellow, minimum
width =18mm, minimum height=8mm]
\tikzstyle{kwc_cs}=[rectangle,draw, trans, minimum
width =12mm, minimum height=6mm]
\tikzstyle{state}=[circle,draw,trans, minimum size=8mm]
\tikzstyle{istate}=[circle,draw, trans, minimum size=8mm]
\tikzstyle{is_state2}=[rectangle,dotted,draw,trans, minimum size=8mm]
\tikzstyle{rstate}=[rectangle,draw,trans, minimum size=10mm,
font=\footnotesize, text width=21mm, align=center]
\tikzstyle{rstate2}=[rectangle,draw,trans, minimum size=10mm,
text width=10mm, align=center]

\def\TO{{\raise .25ex\hbox{\scriptsize\sf to}}}

\def\T{\hbox to 1.5em{\hfill}}

\def\L#1{\raise .2ex\hbox{\scriptsize{$#1$}}&}

\newcommand{\mc}{\textsc{mc2}\xspace}
\newtheorem{replemma}{Lemma}
\newtheorem{reptheorem}{Theorem}
\begin{document}
\title{Certifying Phase Abstraction}
%
%

\author{
Nils Froleyks\inst{1}
 \and
Emily Yu\inst{2} \and
Armin Biere\inst{3} \and
Keijo Heljanko\inst{4} \inst{5}}
\authorrunning{N.~Froleyks et al.}

\institute{
Johannes Kepler University, Linz, Austria \and
Institute of Science and Technology Austria, Klosterneuburg, Austria\and
Albert–Ludwigs–University, Freiburg, Germany \and
University of Helsinki, Helsinki, Finland \and
Helsinki Institute for Information Technology, Helsinki, Finland}
\maketitle              
\begin{abstract}
	Certification helps to increase trust in formal verification
	of safety-critical systems which require assurance on their
	correctness.
	In hardware model checking, a widely used formal verification technique,
	phase abstraction is considered one of the most commonly used
	preprocessing techniques.
	We present an approach to certify an extended form of phase
	abstraction
	using a generic certificate format. As in earlier works our approach involves constructing a
	witness circuit with an inductive invariant property that certifies the correctness of
	the entire model checking process, which is then validated by an independent certificate checker. We
	have implemented and evaluated the proposed approach
	including
	certification
	for various preprocessing configurations  on
	hardware model checking competition benchmarks.
	As an improvement on
	previous work in this area, the proposed method is able to efficiently complete
	certification with an overhead of a
	fraction of model checking time.
\end{abstract}

\section{Introduction}

Over the past few decades, symbolic model
checking~\cite{mcprinciples,mcbook2nd,mchandbook} has been put forward
as one of the
most effective techniques in formal verification.
A lot of trust is placed into model checking tools when assessing the correctness of safety-critical systems.
However, model checkers themselves and the symbolic reasoning tools they rely on,
are exceedingly complex, both in the theory of their algorithms and their practical implementation.
They often run for multiple days, distributed across hundreds
of interacting threads, ultimately yielding a single bit of information signaling the verification result.
To increase trust in these tools, several approaches have attempted to implement fully verified model checkers in a theorem proving environment such as Isabelle~\cite{amjad2003programming,DBLP:conf/cav/EsparzaLNNSS13,sprenger1998verified}.
However, the scalability as well as versatility of those tools is often rather
limited. For example, a technique update tends to require the entire
tool to be re-verified. 

An alternative is to make model checkers provide machine-checkable proofs as certificates that can be validated by independent checkers~\cite{beyer2023bridging,DBLP:conf/sigsoft/0001DDH16,DBLP:journals/tosem/BeyerDDHLT22,DBLP:conf/fmcad/GriggioRT18,DBLP:journals/fmsd/GriggioRT21,DBLP:conf/hvc/KuisminH13,DBLP:conf/fmcad/MebsoutT16,Namjoshi01}, which is already a successful approach in SAT~\cite{heule2021proofs,heule2015proofs}, i.e., proofs are mandatory in the SAT competition since 2016~\cite{SAT-Competition-2016},
and they are a very hot topic in SMT~\cite{DBLP:journals/cacm/BarbosaBCDKLNNOPRTZ23,DBLP:conf/cade/BarbosaRKLNNOPV22,DBLP:conf/smt/HoenickeS22,DBLP:journals/corr/abs-2107-02354} and beyond~\cite{DBLP:journals/cacm/BarbosaBCDKLNNOPRTZ23}.
Crucially, these certificates need to be
simple enough to allow the implementation of a fully verified proof checker~\cite{DBLP:conf/itp/HeuleHKW17,kaufmann2022practical,DBLP:journals/jar/Lammich20},
and preferably verifiable ``end-to-end'', i.e., certifying all stages
of the model checking process, including all forms of preprocessing steps.


The approach in~\cite{biere2022stratified,cav21,fmcad23} introduces a generic
certificate format that can be directly generated from hardware model checkers via book-keeping.
More specifically, the certificate is in the form of a Boolean circuit that comes
with
an inductive invariant, such that it can be verified by six simple SAT checks. 
So far, it
has shown to be effective across several model checking
techniques, but has not covered phase abstraction~\cite{DBLP:conf/iccad/BjesseK05}. The experimental results from~\cite{biere2022stratified,cav21,fmcad23} also show performance challenges with more complex model checking problems. In this paper, we focus on refining the format for smaller certificates
while
accommodating additional techniques such as cone-of-influence analysis
reduction~\cite{mcbook2nd}.

Phase abstraction~\cite{DBLP:conf/iccad/BjesseK05} is a popular
preprocessing technique which tries to simplify a given model checking problem by detecting
and removing periodic signals that exhibit clock-like behaviors. These signals
are essentially the clocks embedded in circuit designs, often due to the design
style of multi-phase clocking~\cite{mony2005exploiting}.
Phase abstraction helps reduce circuit
complexity therefore making the backend model checking task
easier. Differently
from~\cite{baumgartner1999model,baumgartner2003abstraction} where the
concept was first suggested, requiring syntactic analysis and user inputs,
phase abstraction~\cite{DBLP:conf/iccad/BjesseK05} makes use
of
ternary simulation
to automatically identify a group of clock-like latches. Beside this, ternary
simulation has
also been
utilized in the context of temporal
decomposition~\cite{DBLP:conf/fmcad/CaseMBK09} for
detecting transient signals.

In industrial settings, due to the use of complex reset logic as well as circuit
synthesis optimizations, clock signals are sometimes delayed by a number of
initialization
steps~\cite{case2011approximate}.
To further optimize the verification procedure
we extend phase abstraction by exploiting the power of ternary simulation
to capture different classes of periodic signals
including those that are considered partially as clocks as well as equivalent signals~\cite{van1996exploiting}.
An optimal phase number is
computed based on globally extracted patterns, which then is used
to
unfold the
circuit multiple times. The resulting unfolded circuit further undergoes rewriting
and
cone-of-influence reduction, before it is passed on to a base model checker
for final
verification. To summarize our contributions are as follows:
\begin{enumerate}
	\item We formalize, revisit and extend the original phase abstraction~\cite{DBLP:conf/iccad/BjesseK05} by introducing periodic signals, that are then identified and removed for circuit reduction. Our technique also subsumes temporal decomposition~\cite{DBLP:conf/fmcad/CaseMBK09}.
	\item Building upon~\cite{biere2022stratified,cav21,fmcad23}, we propose a refined certificate format
	     for hardware model checking based on a new \textit{restricted simulation} relation. We demonstrate how to build such a certificate for extended phase abstraction.
	\item We present \mc, a certifying model checker that implements our
	      proposed preprocessing technique and generates certificates for the entire model checking process. We show empirically
	      that the approach requires small certification overhead in contrast to~\cite{biere2022stratified,cav21,fmcad23}.
\end{enumerate}

After background
in Section~\ref{sec:background}, Section~\ref{sec:periodicsignal}
introduces the notion of periodic signals. In Section~\ref{sec:pa} we present an extended variant of phase abstraction that simplifies the original model with
periodic signals. In Section~\ref{sec:certification} we define a refined certificate
format and present a general certification approach for phase abstraction. In
Section~\ref{sec:implementation} we describe the implementation of
\mc and then show the effectiveness of our new certification approach
in Section~\ref{sec:exp}.

\section{Background}\label{sec:background}

Given a set of Boolean variables $\mathcal{V}$, a literal $l$ is either
a variable $v\in \mathcal{V}$ or its negation $\neg v$. A
\emph{cube} is considered to be
a non-contradictory set of
literals. Let $c$ be such a cube over a set of variables $L$ and assume $L'$
are copies of $L$, i.e., each $l\in L$ corresponds bijectively to an $l' \in L'$.
Then we write
$c(L')$
to denote the resulting cube after
replacing the variables in $c$ with its corresponding variables in $L'$. For a Boolean
formula $f$, we write $f|_{l}$ and $f|_{\neg l}$ to denote the formula after
substituting all
occurrences of
the
literal $l$ with $\top$ and $\bot$ respectively. We use equality symbols
$\simeq$~\cite{DBLP:books/el/RV01/DegtyarevV01a} and $\equiv$ to denote syntactic
and semantic equivalence
and similarly $\rightarrow$ and $\Rightarrow$ to denote syntactic and semantic
logical implication.


\begin{definition}[Circuit]
	A circuit $C$ is represented by a quintuple $(I,L,R,F,P)$, where $I$
	and
	$L$ are (finite) sets of input and latch variables. The reset functions
	are given as $R=\{r_l(I,L)\mid l\in L\}$
	where the individual reset function
	$r_l(I,L)$
	for a latch $l \in L$
	is a Boolean formula over inputs $I$ and latches $L$. Similarly the set of transition
	functions is given as $F=\{f_l(I,L)\mid l\in L\}$.
	Finally $P(I,L)$ denotes a safety property corresponding to set of
	good states again encoded as a Boolean formula over the inputs and latches.
\end{definition}

This notion can be extended to more general circuits involving for instance
word-level semantics or even continuous variables by replacing in this definition
Boolean formulas
by corresponding predicates and terms in first-order logic modulo
theories.  For simplicity of exposition we focus in this work on Boolean semantics,
which matches the main application area we are targeting, i.e., industrial-scale
gate-level hardware
model checking.  We claim that extensions to ``circuits modulo theories'' are
quite straightforward.

A concrete
state is an assignment to variables $I\cup L$. Therefore the set of reset states of a circuit
is the set of satisfying assignments to
$R(L)=\bigwedge\limits_{l\in
	L}(l\simeq r_l(I,L))$.
\\[-1.5ex]
Note the use of syntactic equality ``$\simeq$'' in this definition.

As in previous work~\cite{biere2022stratified}
we assume acyclic reset functions.
Therefore $R(L)$ is always satisfiable.
A circuit with acyclic reset functions is called \emph{stratified}.

	As in bounded model checking~\cite{DBLP:series/faia/Biere21}, with $I_i$ and $L_i$ ``temporal'' copies of $I$ and $L$ at time step $i$,
the \emph{unrolling} of
a circuit up to length $k$ is expressed as:
$$U_k=\kern-.5em\bigwedge\limits_{i\in[0,k)}(L_{i+1}\simeq
	F(I_i,L_i)).$$ 

Cube simulation~\cite{fmcad23} subsumes ternary simulation such that a
lasso found by ternary simulation can also be found via cube simulation. A cube
simulation is a sequence of cubes $c_0,{\ldots},c_\delta, {\ldots},c_{\delta+\omega}$
over latches $L$ such that (1)~$R(L)\Rightarrow c_0$; (2) $c_i \wedge
	(L'\simeq F(I,L))\,\Rightarrow\,
	c'_{i+1}$ for all $i\in[0,\delta+\omega)$, where $c'_{i+1}$ is the primed copy of
$c_{i+1}$. It is called
a cube lasso if $c_{\delta+\omega}\wedge
	(L'\simeq F(I,L))\,\Rightarrow\,
	c'_{\delta}$.
	In which case $\delta$ is the stem length and $\omega$ is the loop length.
	For $\delta = 0$, the initial cube is already part of the loop and for $\omega = 0$,
	the lasso ends in a self-loop.
\section{Periodic Signals}\label{sec:periodicsignal}

In sequential hardware designs, signals that eventually stabilize to a constant, i.e.,
to $\top$ or $\bot$, after certain
initialization steps are called \emph{transient}
signals~\cite{DBLP:conf/fmcad/CaseMBK09,fmcad23},
whereas
oscillating signals have clock-like or periodic behaviors. A
simplest example of a clock is a latch that always oscillates between $\top$ and $\bot$.

Since
hardware designs typically consist of complex initialization logic, there are
occurrences of delayed oscillating signals, like clocks that start ticking after several reset steps, with a combination of transient and clock behaviours. We generalize this concept to
categorize latches as periodic
signals associated with a \textit{duration} (i.e., the number of time steps for
which a
signal is delayed) and a \emph{phase number} (i.e., the period length
in a periodic behavior). Moreover, our
generalization also captures equivalent and antivalent
signals~\cite{van1996exploiting}, as well as those that exhibit partial periodic behaviours. See Fig.~\ref{fig:sig} for an example.
\begin{figure*}
	\centering\vspace{-5pt}
	\includegraphics[width=\textwidth]{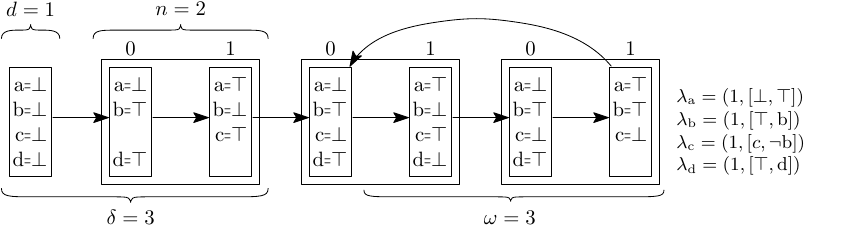}
	\caption{
		An example of a cube lasso over the latches $l \in L = \{a, b, c, d\}$.
		In the example the tall rectangles represent cubes as
		partial assignments, i.e., the second cube from the left
		is $(\neg a) \wedge b \wedge d$.
		Phase 0 and 1 are marked on top of the cubes.
		As shown, duration $d=1$ and phase number $n=2$ yield a high number of useful signals for this cube lasso.
		Each latch $l$ is associated with a periodic pattern $\lambda_l$ on the right describing its behaviors for phase 0 and 1.
		If a latch is missing from a cube for a certain phase, it has no constraint thus we
		use the equality of the latch to itself in the signal.
		Latch $a$ turns out to be a simple clock delayed by one step.
		Latches $b$ and $d$ behave clock-like but only in phase 0.
		Latch $c$ always has the opposite value of latch $b$ in phase 1.
		Note that we could also have $\neg c$ in phase 1 of signal $\lambda_{b}$ but
		choosing a single representative for a set of equivalent signals is beneficial for the
		following simplification steps.
	\label{fig:sig}
	}
\end{figure*}
\begin{definition}[Periodic Signal]\label{def:ps}
	Given a circuit $C=(I,L,R,F,P)$ and a cube lasso $c_0,{\ldots}c_\delta,
		{\ldots},c_{\delta+\omega}$. A periodic signal $\lambda_l$ for a latch $l\in L$ is
	defined as
	$\lambda_l=(d,[v^0,{\ldots},v^{n-1}])$ where $d\in \mathbb{N},$ $n\in
		\mathbb{N^+}$
	and $v^i$ is
	a latch
	literal or a constant, with
	$d\leq \delta$.
	We further require that there exist
	$k^0,k^1\in\mathbb{N^+}$ with $k^0\cdot n+d=\delta$ and
	$k^1\cdot n=\omega+1$ such that
	for all $i \in [0,n)$ and
	$j \in [0,k^{0}+k^{1})$ we have
	$c_{i+j\cdot n} \Rightarrow (l\simeq v^{i})$.
\end{definition}

For a signal $\lambda_l = (d,[v^0,{\ldots},v^{n-1}])$
we will write $\lambda_{l}^i$ to refer to
the $i$-th element of $[v^0,{\ldots},v^{n-1}]$, which we refer to as its phase.
See Fig.~\ref{fig:sig} for an example where $k^0 = 1$ and $k^1=2$.

\begin{figure*}
	\centering
	\includegraphics[width=\textwidth]{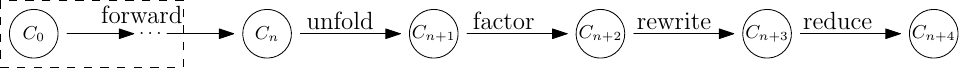}
	\caption{Circuit transformation using phase abstraction. }
	\label{fig:pa}
\end{figure*}

\section{Extending Phase Abstraction}
\label{sec:pa}

In this section, we revisit and extend phase abstraction by defining it as a sequence of
preprocessing
steps, as illustrated in Fig.~\ref{fig:pa}. Differently
from the approach in~\cite{DBLP:conf/iccad/BjesseK05}, we present phase
abstraction as part
of a compositional framework, that handles a more general class of periodic
signals. 
As our approach subsumes temporal decomposition adopted from the
framework in~\cite{fmcad23}, we first apply \emph{circuit
forwarding}~\cite{fmcad23} for duration $d$ (i.e., unrolling the reset states of a circuit by $d$ steps) before unfolding is performed.

Fig.~\ref{fig:pa} illustrates the flow of
phase
abstraction. The process begins by using cube simulation to identify
a set of periodic signals as defined in Section~\ref{sec:periodicsignal} and
computing an optimal duration and phase number based on a
selected cube lasso.
Once the circuit is unfolded $n$ times, factoring is performed by assigning
constant values to the clock-like signals as well as replacing latches with their
equivalent or antivalent
representative latches in each phase. Next, the
property is rewritten by applying structural rewriting techniques and the
circuit is further simplified using cone-of-influence reduction. Finally, the simplified
circuit ($C_{n+4}$ in Fig.~\ref{fig:pa}) is checked using a base model checking
approach such as
IC3/PDR~\cite{DBLP:conf/vmcai/Bradley11} or continue to be preprocessed further. 

In Fig.~\ref{fig:unfolded}, we show intuitively an example of a circuit with 4-bit states representing 0,...,9 and so on, where the initial state is 0. After forwarding the circuit by one step ($d=1$), the initial state becomes 1. Subsequently with an unfolding of $n=3$, every state (marked with rectangles) in the unfolded circuit consists of three states from the original circuit. We introduce the formal definitions below.

Unfolding a circuit simply means to copy the transition function multiple times to compute $n$ steps of the original circuit at once.
Each copy of the transition function only has to deal with a single phase and can therefore be optimized.

\begin{definition}[Unfolded circuit]\label{def:unfold}
	Given a circuit $C=(I, L, R, F, P)$ and a phase number
	$n\in\mathbb{N^+}$. The unfolded circuit $C'=(I',L',R',F',P')$ is:
	\begin{enumerate}
		\item $I'=I^0\cup\cdots\cup I^{n-1}$;  $L'=L^0\cup\cdots\cup L^{n-1}$.
		\item $R'=\{r'_l\mid l\in L'\}: $ for $l\in L^0, r'_l=r_l$;\\
		      \phantom{$R'=\{r'_l\mid l\in L'\}: $ }for $i\in (0,n), l^i\in L^i,r'_{l^i} = F(I^i, L^{i-1}).$
		\item $F'=\{f'_l\mid l\in L'\}:$ for $l\in L^0, f'_l=f_l(I^0,L^{n-1})$;\\
		      \phantom{$F'=\{f'_l\mid l\in L'\}:$ }for $i\in(0,n), l^i\in L^i, f'_{l^i} = f_{l^i}(I^i, L^{i-1}).$
		\item $P' = \bigwedge\limits_{i\in[0,n)} P(I^i,L^i).$
	\end{enumerate}
\end{definition}

\begin{figure}[t]
	\centering
	\includegraphics{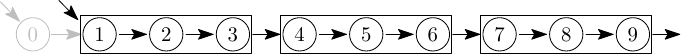}
	\caption{
		An example of a forwarded ($d=1$) and unfolded ($n=3$) circuit.
		The circles denote states in the original circuit ($0$ is the initial state).
		The rectangle are states in the unfolded circuit.
	}
	\label{fig:unfolded}
\end{figure}

We obtain a simplified circuit by replacing the periodic signals with
constants and equivalent/antivalent latches in the unfolded circuit.

\begin{definition}[Factor circuit]\label{def:fac}
	For a fixed duration $d$ and phase number $n$,
	given a $d$-forwarded and $n$-unfolded circuit $C=(I, L,R, F, P)$
	and a periodic signal with duration $d$ and phase number $n$ for each latch,
	the factor circuit
	$C'=(I, L, R', F', P)$
	is defined by:

	\vspace{4pt}
	\begin{minipage}{.45\textwidth}
		$R'=\{r'_l\mid l\in L\}:$
		\vspace{-6pt}
		\begin{itemize}
			\item $r'_{l^{i}} = \lambda^{i}_{l}, \text{ if } \lambda^{i}_{l} \in \{\bot, \top\};$
			\item $r'_{l^{i}} = r_{\lambda^{i}_{l}}, \text{if } \lambda_l^i\in L.$
			\item $r'_{l^{i}} = \neg r_{\neg\lambda^{i}_{l}}, \text{otherwise.} $
		\end{itemize}
	\end{minipage}%
	\begin{minipage}{0.5\textwidth}
		$F'=\{f'_l\mid l\in L\}:$
		\vspace{-6pt}
		\begin{itemize}
			\item $f'_{l^{i}} = \lambda^{i}_{l}, \text{ if } \lambda^{i}_{l} \in \{\bot, \top\};$
			\item $f'_{l^{i}} = f_{\lambda^{i}_{l}}, \text{if } \lambda_l^i\in L.$
			\item $f'_{l^{i}} = \neg f_{\neg\lambda^{i}_{l}}, \text{otherwise.} $
		\end{itemize}
	\end{minipage}
\end{definition}

Replaced latches will be removed by a combination of rewriting and cone-of-influence reduction introduced
in the following definitions.
There are various rewriting techniques
also including SAT
sweeping~\cite{fujita2015toward,kuehlmann2002robust,mishchenko2006dag,mishchenko2006improvements,mishchenko2005fraigs,zhu2006sat}.

\begin{definition}[Rewrite circuit]
	Given a circuit $C=(I, L, R, F, P)$, a rewrite circuit $C'=(I, L, R, F, P')$
	satisfies $P\equiv P'$.
\end{definition}

For a given circuit, we apply cone-of-influence reduction to obtain a reduced
circuit such that latches and inputs outside the cone of influence are
removed.

\begin{definition}[Reduced circuit]\label{def:reduced}
	Given a circuit $C=(I, L, R, F, P)$. The
	reduced circuit $C'=(I',L',R',F',P)$ is
	defined as follows:\vspace{0.6em}
	\begin{itemize}

		\begin{minipage}{.4\textwidth}

			\item $I'=I\cap coi(P);$
			\item $R'=\{r_l\mid l\in L'\}$;

		\end{minipage}%
		\begin{minipage}{0.5\textwidth}

			\item $L'= L\cap coi(P)$;

			\item $F'=\{f_l\mid l\in L'\}$, 

		\end{minipage}
	\end{itemize}
	where the cone of influence of the property $coi(P) \subseteq (I \cup L)$
	is defined as the smallest set of inputs and latches such that
	$vars(P)\subseteq coi(P)$
	as well as
	$\/vars(r_l)\subseteq coi(P)$ and
	$vars(f_l)\subseteq coi(P)$
	for all latches $\/l\in coi(P)$.
\end{definition}
\section{Certification}
\label{sec:certification}
We define a
revised certificate format that allows smaller and more
optimized certificates. We then propose a
method for producing certificates for phase abstraction.
The proofs for our main theorems can be found in the Appendix.

\subsection{Restricted simulation}

In the following, we define a new variant of the stratified simulation
relation~\cite{biere2022stratified}, which we call \emph{restricted
	simulation}, that considers the intersection of latches shared
between
two
given circuits as a common component.

\begin{definition}[Restricted Simulation]\label{def:ressim}
	Given stratified circuits $C'$ and $C$ with
	$C'=(I',L',R',F',P')$ 
	and
	$C=(I,L,R,F,P)$.
	We say $C'$
	\emph{simulates} 
	$C$
	under the \emph{restricted simulation relation} iff
	\begin{enumerate}
		\item\label{def:ressim:reset}For $l\in (L\cap L'), r_l(I,L)
			\equiv
			r'_l(I',L').$
		\item\label{def:ressim:trans} For $l\in  (L\cap L'), f_l(I,L) \equiv
			f'_l(I',L').$
		\item\label{def:ressim:prop} $P'(I',L')\Rightarrow P(I,L).$

	\end{enumerate}
\end{definition}

This new simulation relation differs from~\cite{biere2022stratified,cav21}, where
inputs were required to be
identical in
both circuits ($I=I'$), and latches in $C$ had to form a subset of latches in $C'$
($L\subseteq L'$).  Therefore, under those previous ``combinational''~\cite{cav21} or ``stratified''
\cite{biere2022stratified} simulation relations the simulating circuit
$C'$ cannot have fewer latches than $L$. This is a feature we need
for instance when incorporating certificates for cone-of-influence
reduction~\cite{mcbook2nd}, a common preprocessing technique.
It opens up the possibility to reduce certificate sizes substantially.

Still, as for stratified simulation,
restricted simulation can be verified by three simple SAT checks,
i.e., separately for
each of the three requirements in~Def.~\ref{def:ressim}.

\begin{definition}[Semantic independence]
	Let $\mathcal{V}$ be a set of variables and $v \in \mathcal{V}$. Then a formula
	$f(\mathcal{V})$ is said
	to
	be semantically independent of $v$ iff
	$$f(\mathcal{V})|_{v}\equiv
		f(\mathcal{V})|_{\neg v}.$$
\end{definition}

Semantic
dependency~\cite{fleury2023mining,lagniez2020definability,padoa1901essai,slivovsky2020interpolation}
allows us to assume that a formula only depends on a
subset of variables, which without loss of generality simplifies proofs used for
the rest of this section. The stratified assumption for reset functions entails no
cyclic dependencies
thus $R'(L')$ is satisfiable. A reset state in a circuit is simply a satisfying
assignment to
the reset predicate $R(L)$. Based on the reset
condition (Def.~\ref{def:ressim}.\ref{def:ressim:reset}), it is however
necessary to
show
that for every reset state in $C$ it can always be extended to a reset state in
$C'$, while the common variables have the same assignment in both
circuits. This is stated in the lemma below, and the proofs can be found in the Appendix. 

\begin{lemma}\label{lem:konly}
	Let $C=(I,L,R,F,P)$ and $C'=(I',L',R',F',P')$
	be two stratified circuits 
	satisfying the reset condition defined in
	Def.~\ref{def:ressim}.\ref{def:ressim:reset}. Then
	$R'(L\cap L')$ is semantically dependent only on their common variables.
\end{lemma}

In fact, semantic independence is a direct consequence of restricted simulation;
thus no separate check is required.
We make a further remark
that if
the reset function is dependent on an input
variable, then it has to be an input variable common to both circuits.

Based on this, we conclude with
the main theorem for restricted simulation such that $C$ is safe if $C'$ is safe (i.e., no bad state that violates the property is reachable from any initial state).

\begin{theorem}\label{thm:safety}
	Let $C=(I,L,R,F,P)$ and
	$C'=(I',L',R',F',P')$ be two stratified circuits, where $C'$ simulates $C$
	under restricted simulation. \\If
	$C'$ is
	safe, then $C$ is also safe.
\end{theorem}

Intuitively, if there is an unsafe trace in $C$, Def.~\ref{def:ressim}.\ref{def:ressim:reset} together with Lemma~\ref{lem:konly}
allow us to find a simulating reset state and transition it with Def.~\ref{def:ressim}.\ref{def:ressim:trans} to a simulating state also violating the property in $C'$ by Def.~\ref{def:ressim}.\ref{def:ressim:prop}.
Here a state in $C'$ simulates a state in $C$ if they match on all common variables.
Building on this, we present
witness circuits as a format for certificates.
Verifying the restricted simulation relation requires three SAT checks, and another three SAT checks are needed for validating the inductive invariant~\cite{cav21}. Therefore certification requires in total six SAT checks as well as a
polynomial time check for reset stratification. 

\begin{definition}[Witness circuit]\label{def:wit}
	Let $C=(I,L,R,F,P)$ be a stratified circuit. A witness circuit $W=(J,M,S,G,Q)$
	of $C$ satisfies the following:
	\begin{itemize}
		\item $W$ simulates $C$ under the restricted simulation relation.
		\item $Q$ is an inductive invariant in $W$.
	\end{itemize}
\end{definition}

The witness circuit format subsumes~\cite{biere2022stratified,fmcad23}, thus
every witness circuit in their format is also valid under
Def.~\ref{def:wit}.

\begin{figure}
	\centering
	\includegraphics[width=\textwidth]{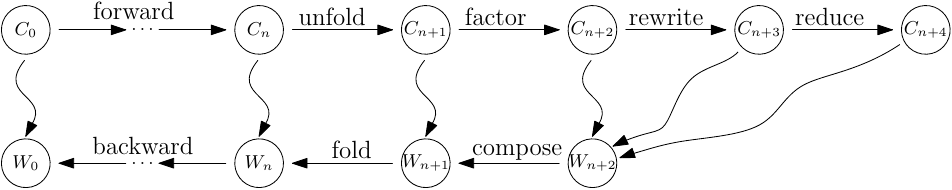}
	
	\caption{Certification for (extended) phase abstraction. Base model checking is
		performed on
		circuit $C_{n+4}$, which produces a witness circuit $W_{n+2}$,
		that certifies $C_{n+2}, C_{n+3}, \text{ and } C_{n+4}$.
		We construct
		step-wise to
		obtain
		$W_0$, which is a certificate for the entire model checking
		procedure.\label{fig:pacert}}
\end{figure}

\subsection{Certifying Phase Abstraction}

The certificate format is generic, subsumes~\cite{fmcad23}, and is designed to potentially be used as a standard in future hardware model checking competitions. We proceed to demonstrate how a certificate can be constructed for a model checking pipeline that includes phase abstraction. The theorems in this section state that this construction guarantees that a certificate will be produced.
We illustrate our certification pipeline in Fig.~\ref{fig:pacert}. After phase
abstraction and base model checking, we can build a certificate backwards
based on the certificate produced by the base model checker. 
The following theorem states that the witness
circuit of the reduced circuit serves as a witness circuit for the original circuit
too.

\begin{theorem}\label{thm:reduced}
	Given a circuit $C=(I,L,R,F,P)$ and its reduced circuit
	$C'=(I',L',R',F',P')$. A witness circuit of $C'$ is
	also a
	witness circuit of $C$.
\end{theorem}

The outcome of rewriting is a
circuit with a simplified property that maintains semantic equivalence with the
original
property. Therefore in our framework, the certificate for the simplified property
is also
valid for the original property.  Furthermore, certificates can be optimized by rewriting at any stage.
We summarize this in the following
proposition.

\begin{proposition}
	Given a circuit $C$ and its rewrite circuit $C'$. A witness circuit of $C'$ is
	also a
	witness circuit of $C$.
\end{proposition}

We define the composite witness circuit to combine the certificates
for cube
simulation and the factor circuit.

\begin{definition}[Composite witness circuit]\label{def:compwit}
	Given a stratified circuit $C=(I,L,R,F,P)$ and its factor circuit
	$C'=(I',L',R',F',P')$, and the
	unfolded loop invariant
	$\phi = \vee_{i\in[0,m)}\wedge_{j\in[0,n)}c_{i*n+j+d}$,
	with $m=(\delta+\omega-d+1)/n$,
	obtained from the cube lasso.
	Let
	$W'=(J',M',S',G',Q')$ be a
	witness circuit
	of $C'$. The composite witness circuit $W=(J,M,S, G,Q)$ is defined as
	follows:\vspace{0.6em}
	\begin{enumerate}
		\begin{minipage}{.52\textwidth}
			\item $J=I\cup J'$.
			\item $M=L\cup (M'\backslash L')$.
			\item $S=\{s_l\mid l\in M\}$:\vspace{-7pt}
			\begin{enumerate}
				\item for $l\in L, s_l = r_l$;
				\item for $l\in M'\backslash L', s_l=s'_l.$
			\end{enumerate}
		\end{minipage}%
		\begin{minipage}{0.6\textwidth}
			\item $G=\{g_l\mid l\in M\}$:\vspace{-7pt}
			\begin{enumerate}
				\item for $l\in L, g_l = f_l$;
				\item  for $l\in M'\backslash L', g_l=g'_l.$
			\end{enumerate}

			\item $Q=\phi(L)\wedge Q'(J',M').$
		\end{minipage}
	\end{enumerate}
\end{definition}

\begin{theorem}~\label{thm:compwit}
	Given circuit $C=(I,L,R,F,P)$, and factor circuit $C'=(I',L',\\R',F',P')$. Let
	$W'=(J',M',S',G',Q')$ be a witness circuit of $C'$, and $W=(J,M,S, G,Q)$
	constructed as in Def.~\ref{def:compwit}. Then $W$ is a witness
	circuit of $C$.
\end{theorem}
\begin{figure}
	\centering
	\includegraphics[width=\textwidth]{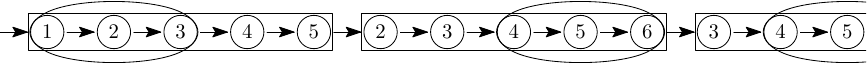}
	\caption{
		Every fully initialized state of a 3-folded witness circuit
		contains 3 original states that form an unfolded state.
		Two consecutive 3-folded states contain either the same unfolded states
		or two states consecutive in the unfolded circuit.
		\label{fig:fold}
	}
\end{figure}

In the construction of an $n$-folded witness circuit from the unfolded witness $W'$, a single instance of $W'$'s latches ($N$), yet multiples of the original latches $L$ are used.
As illustrated in Fig~\ref{fig:fold},
these $L$ record a history, contrasting with their role in the unfolded circuit where they calculate multi-step transitions.

\begin{definition}[$n$-folded witness
		circuit]\label{def:unfoldwit}
	Given a circuit $C=(I, L, R, F, P)$ with a phase number $n\in \mathbb{N^+}$,
	and its
	unfolded circuit
	$C'=(I',L',R',F',P')$. Let $W'=(J',M',S',G',Q')$ be the witness circuit of $C'$. The
	$n$-folded witness circuit $W=(J,M,S,G,Q)$ is defined as follows:
	\begin{enumerate}
		\item $J=I^0\cup J^0$, where $I^0$ and $J^0$ are $I$
		      and $J'$ respectively.
		\item $M=I^1\cdots I^m\cup L^0\cdots L^m \cup N \cup J^1
			      \cup
			      \{b^0\cdots b^m,e^0\cdots e^{n-2}\},$
		      \\where $m=2\times n-2$,
		      $N=M'\setminus L'$,
		      and $I^i,L^i$
		      are
		      copies of $I$ and $L$, and $J^1$ is a copy of $J'$.
		\item $S=\{s_l\mid l\in M\}$:
		      \begin{enumerate}
			      \item $s_{b^0}=\top;$

			      \item For $i\in (0,m], s_{b^i}=\bot.$
			      \item For $i\in [0,n-1), s_{e^i}=\bot.$
			      \item For $l\in L^0, s_l=r'_l.$
			      \item For $l\in (I^1\cdots I^m\cup L^1\cdots L^m\cup J^1),
				            s_l=l.$
			      \item For $l\in N, s_l=s'_l.$

		      \end{enumerate}

		\item $G=\{g_l\mid l\in M\}$:
		      \begin{enumerate}
			      \item $g_{b^0}=\top$.
			      \item For $i\in [1,m], g_{b^i}= b^{i-1}.$
			      \item $g_{e^0}=b^{n-1}\wedge \neg e^{n-2}$.
			      \item For $i\in[1,n-1), g_{e^i}=e^{i-1}\wedge\neg
				            e^{n-2}$.
			      \item For $l\in L^0, g_l=f_l.$
			      \item For $l^1\in J^1,g_{l^1}=l^0.$
			      \item For $i\in[1,m]$, $l^i \in (I^i\cup L^i),
				            g_{l^i}=l^{i-1}.$
			      \item For $l\in N, \\g_l=ite(
				            \begin{aligned}[t]
					             & e^{n-2},
					            g'_l(J^1,M'\cap(I^{m-n+1}\cdots I^m\cdots L^{m-n+1}\cdots
					            L^m\cup N)),
					            l).
				            \end{aligned}$
		      \end{enumerate}
		\item $Q=\bigwedge\limits_{i\in[0,6]}q^i:$
		      \begin{enumerate}

			      \item $q^0 = P(I^0, L^0).$
			      \item $q^1=b^0.$
			      \item $q^2= \bigwedge\limits_{i\in[1,m]} (b^i\rightarrow
				            b^{i-1})$.
			      \item $q^3= \bigwedge\limits_{i\in[1,m]} (b^i\rightarrow
				            (L^i\simeq
				            F(I^{i-1},L^{i-1}) ))$.
			      \item $q^4= \bigwedge\limits_{i\in[1,m]}((\neg b^i\wedge
				            b^{i-1})\rightarrow (R(L^{i-1})\wedge
				            S'(N)))$.
			      \item $q^5=b^m\rightarrow
				            \begin{aligned}[t]
					             & (\bigvee\limits_{i\in[0,n)}
					            ((\bigwedge\limits_{j\in[i,n-1)}\neg e^j)\wedge
					            (\bigwedge\limits_{j\in[0,i)} e^j)\wedge \\
					             & Q'(J^0,M'\cap(L^i\cdots\cup
					            L^{i+n-1}\cup N)) ).
				            \end{aligned}$
			      \item $q^6=\bigwedge\limits_{i\in[1,n-2]} (e^i\rightarrow
				            e^{i-1})$.
			      \item $q^7=\bigwedge\limits_{i\in [0,n-2] }(e^i\rightarrow b^{n+i})$.
			      \item $q^8 = \bigwedge\limits_{i\in [0,n-2)} ((\neg b^m\wedge
				            b^{n+i})\rightarrow e^i)$.
		      \end{enumerate}
	\end{enumerate}
\end{definition}

The $b^{i}$s are used for encoding initialization.
So that inductiveness is ensured when not all copies are initialized.
The $n-1$ bits $e^{i}$ are used to determine which set of $n$
consecutive original states form an unfolded state (a state in the unfolded circuit).
This information is used to
determine on which copies the unfolded property needs to hold and to transition
the latches in $N$ (the part of the witness circuit added by the backend model checker) once every $n$
steps.

\begin{theorem}\label{thm:nfolded}
	Given a circuit $C=(I, L, R, F, P)$ with a phase number $n\in\mathbb{N^+}$,
	its
	unfolded cicuit
	$C'=(I',L',R',F',P')$ with a witness circuit $W'=(J',M', S', G', Q')$.
	Let $W=(J,
		M, S, G, Q)$ be the circuit constructed as in
	Def.~\ref{def:unfoldwit}. Then $W$ is a witness circuit of $C$.
\end{theorem}

After the witness circuit has been folded, the same construction from~\cite{fmcad23} can be used to construct the backward witness. With that, the pipeline outlined in Fig.~\ref{fig:pacert} is completed. If phase abstraction is the first technique applied by the model checker, a final witness is obtained. Otherwise, further witness processing steps still need to be performed.

\begin{figure}
	\centering
	\includegraphics[width=\textwidth]{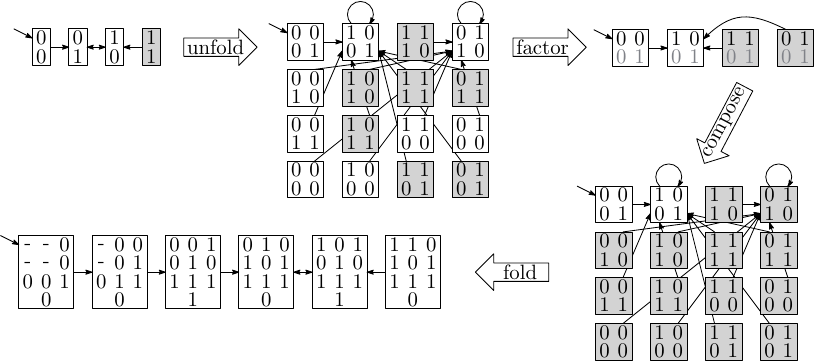}
	\caption{
A concrete example of the model checking and certification pipeline.
The original circuit has two latches; the bottom latch alternates and the top copies the previous value of this clock. The property is that at least one bit is unset. Bad states are marked gray. After unfolding with phase number two, the size of the state space is squared. Since the bottom bit is periodic, we can replace it with a constant in each phase (factor). On this circuit terminal model checking is performed, since the property is already inductive (no transition from good to bad), the circuit serves as its own witness.
To produce the final witness circuit, the clock is added back as a latch, and the property is extended with the loop invariant asserting that the clock has the correct value for each phase.
Lastly, the circuit is \emph{folded} to match the speed of the original circuit.
Three initialization bits $b^{i}$ are introduced and one additional bit $e^{0}$ that determins which pair of consecutive states need to fulfill the property ($0$ for the right pair and $1$ for the left).
This check is only part of the property once full initialization is reached.
For this final witness circuit, only the good states are depicted.
Also, the first two states represent sets of good states with the same behavior.
		\label{fig:example}
	}
\end{figure}

\section{Implementation}
\label{sec:implementation}
In this section, we present \mc, a certifying model checker implementing phase
abstraction and IC3.
We implement our own IC3 since no existing model checker supports reset functions
or produces certificates in the desired format.
We used fuzzing to increase trust in our tool. The version of \mc used for the evaluation,
was tested on over 25 million randomly generated circuits~\cite{Biere-FMV-TR-11-2} in combination
with random parameter configurations. All produced certificates where verified.

To extract periodic signals we perform ternary
simulation~\cite{DBLP:journals/fmsd/SegerB95} while using a forward-subsumption algorithm based on
a one-watch-literal data structure~\cite{onewatchscheme} to identify supersets of previously visited cubes, and
thereby a set of cube lassos.
For each cube lasso we consider every factor of the loop length $\omega$ as a phase number candidate $n$.
We also consider every duration $d$, that renders the leftover tail length ($\delta - d$) divisible by $n$.
To keep the circuit sizes manageable, we limit both $n$ and $d$ to a maximum of $8$.
We call each pair ($d$, $n$) an unfolding candidate and compute the corresponding periodic signal (Def.~\ref{def:ps}) for each latch.

For each phase,
equivalences are identified by inserting a bit string corresponding to the signs of each latch into a hash table.
After identifying the signals,
forwarding and unfolding are performed on a copy of the circuit, followed by
rudimentary rewriting.
Currently the rewriting does not include structural hashing and is mostly limited to constant propagation.
Afterwards a sequential cone-of-influence analysis starting from the property is performed.
After performing these steps for each candidate, we pick the
duration-phase pair that
yields a circuit with the fewest latches and give it to a backend model checker.

We evaluated the preprocessor on three backend model checkers:
the open-source $k$-induction-based model checker McAiger~\cite{mcaiger}(Kind in the following),
the state-of-the-art IC3 implementation in ABC~\cite{DBLP:conf/cav/BraytonM10}
and our own version of IC3 that supports reset functions
and produces certificates.
Since ABC does not support reset functions, it is not able to model check
any forwarded circuit (note that implementing this feature on ABC is also a
non-trivial
task),
therefore for this configuration we only ran phase abstraction without
forwarding thus no temporal decomposition.

Our IC3 implementation on \mc does feature cube minimization via
ternary simulation~\cite{een2011efficient}, however it is missing
proof-obligation rescheduling.
In fact, we currently use a simple stack of proof obligations as opposed to a priority queue.
Despite using one SAT solver instance per frame, we also do not feature cones-on-demand,
but instead always translate the entire circuit using
Tseitin~\cite{tseitin1983complexity}.

Lastly, we also modified the
open source implementation of Certifaiger~\cite{certifaiger} to support
certificates based on restricted simulation.
For a witness circuit $C'$ of $C$, the new certificate checker encodes the following six checks as combinatorial AIGER circuits and then uses the \texttt{aigtocnf} to translate them to SAT:
\begin{enumerate}[itemsep=0pt,label=\protect\mycirc{\Alph*}]
\item The property of $C'$ holds in all initial states.
\item The property of $C'$ implies the property for successor states.
\item The property of $C'$ holds in all good states.
\item The reset functions of common latches are equivalent. (Def.~\ref{def:ressim}.~\ref{def:ressim:reset})
\item The transition functions of common latches are equivalent. (Def.~\ref{def:ressim}.~\ref{def:ressim:trans})
\item The property of $C'$ implies the property of $C$. (Def.~\ref{def:ressim}.~\ref{def:ressim:prop})
\end{enumerate}
The first three checks are unchanged and encode the standard check for $P'$ being an inductive invariant in $C'$.
Since $P'$ is both the inductive invariant and the property we are checking, \protect\mycirc{C} can technically be omitted.
However, in our implementation, the inductiveness checker is an independent component from the simulation checker and would
also works for scenarios where the inductive invariant is a strengthening of the property in $C'$.

\begin{table}[h!]

	\begin{tabular}{@{}lrrrrc@{\hspace{-3pt}}r@{\hspace{-1pt}}rc@{\hspace{-1pt}}rrrc@{\hspace{-1pt}}rrr@{}}
		\toprule
		                        & \multicolumn{4}{c}{Model} & \phantom{x}              &
		\multicolumn{2}{c}{ABC} & \phantom{x}               & \multicolumn{3}{c}{Our
		IC3}                    & \phantom{x}               & \multicolumn{3}{c}{Kind}                                                                                      \\
		\cmidrule{2-5} \cmidrule{7-8} \cmidrule{10-12} \cmidrule{14-16}
		                        & {\hspace{-3cm}Safe}       & {$\overline{n}$}         & {$d$} & {$n$} & {} & {} & {PA} & {} & {} & {PA} & {Full} & {} & {} & {PA} & {Full}
		\\
		\midrule
		Solved &  &  &  &  &  & 740 & \bt 745 &  & 715 & 715 & 604 &  & 533 & 538 & 544 \\
PAR2 &  &  &  &  &  & 996 & \bt 941 &  & 1357 & 1351 & 2699 &  & 3533 & 3472 & 3399 \\
\midrule
abp4p2ff & \checkmark & 1 & 1 & 1 &  & 1.12 & \bt 1.08 &  & 6.35 & 6.23 & 6.18 &  & 2.50 & 2.50 &  \\
bjrb07 &  & 3 & 0 & 3 &  & 0.12 & 0.10 &  & 0.11 & \bt 0.03 & 0.07 &  &  & \bt 0.03 & 0.04 \\
nusmvb5p2 &  & 5 & 0 & 5 &  & 0.12 & 0.10 &  & \bt 0.01 & \bt 0.01 & \bt 0.01 &  &  & \bt 0.01 & \bt 0.01 \\
nusmvb10p2 &  & 5 & 0 & 5 &  & 0.22 & 0.13 &  & 0.10 & 0.03 & 0.04 &  &  & \bt 0.02 & \bt 0.02 \\
prodcell0 & \checkmark & 1 & 5 & 8 &  & 26.97 & 27.07 &  & 228.46 & 243.73 & 49.76 &  &  &  & \bt 2.37 \\
prodcell0neg & \checkmark & 1 & 5 & 8 &  & 16.36 & 15.93 &  & 230.57 & 230.67 & 36.62 &  &  &  & \bt 2.39 \\
prodcell1 & \checkmark & 1 & 7 & 8 &  & 23.45 & 23.38 &  & 654.21 & 665.86 & 59.67 &  &  &  & \bt 4.43 \\
prodcell1neg & \checkmark & 1 & 7 & 8 &  & 28.36 & 28.33 &  & 681.11 & 738.61 & 61.74 &  &  &  & \bt 4.48 \\
prodcell2 & \checkmark & 1 & 7 & 8 &  & 24.98 & 24.58 &  & 661.71 & 663.37 & 56.74 &  &  &  & \bt 4.43 \\
prodcell2neg & \checkmark & 1 & 7 & 8 &  & 20.23 & 20.28 &  & 778.39 & 768.75 & 56.14 &  &  &  & \bt 4.47 \\
bc57sen0neg & \checkmark & 1 & 1 & 1 &  & 503.61 & \bt 494.55 &  & 910.72 & 906.87 & 1760.41 &  &  &  & 830.92 \\
abp4ptimo & \checkmark & 1 & 1 & 1 &  & 4.14 & \bt 4.13 &  & 28.93 & 29.91 & 6.32 &  &  &  & 608.55 \\
boblivea &  & 1 & 2 & 1 &  & 3.70 & \bt 3.68 &  &  &  & 7.85 &  &  &  &  \\
bobsm5378d2 &  & 1 & 8 & 1 &  & \bt 4.04 & 4.12 &  &  &  & 88.36 &  &  &  &  \\
bobsmnut1 &  & 8 & 5 & 8 &  & \bt 10.95 & 40.08 &  &  &  & 2504.07 &  &  &  &  \\
prodcell3neg & \checkmark & 2 & 2 & 8 &  & 27.88 & 10.86 &  & 310.22 &  &  &  &  & 837.43 & \bt 2.73 \\
prodcell4neg & \checkmark & 2 & 2 & 8 &  & 44.31 & 9.90 &  & 404.12 &  &  &  &  & 26.04 & \bt 2.77 \\
prodcell3 & \checkmark & 2 & 2 & 8 &  & 23.45 & 11.24 &  & 320.23 &  &  &  & 1103.29 & 19.22 & \bt 2.48 \\
prodcell4 & \checkmark & 2 & 2 & 8 &  & 31.40 & 10.08 &  & 398.83 &  &  &  & 29.71 & 18.67 & \bt 2.68 \\
pdtvisvsar29 &  & 1 & 2 & 5 &  &  &  &  &  &  & 1523.73 &  & 0.36 & \bt 0.29 & 0.40 \\
intel042 & \checkmark & 1 & 3 & 2 &  &  &  &  &  &  &  &  & \bt 3876.04 & 4061.38 &  \\
intel022 &  & 2 & 2 & 2 &  &  & \bt 1852.29 &  &  &  &  &  &  &  &  \\
intel021 &  & 2 & 2 & 2 &  &  & 2752.86 &  &  & \bt 651.56 &  &  &  &  &  \\
intel023 &  & 2 & 2 & 2 &  &  & \bt 2257.94 &  &  & 3728.38 &  &  &  &  &  \\
intel029 &  & 2 & 2 & 2 &  &  & \bt 2550.14 &  &  & 3437.64 &  &  &  &  &  \\
intel024 &  & 2 & 2 & 2 &  &  & \bt 167.96 &  & 676.64 & 4526.60 &  &  &  &  &  \\
intel019 &  & 2 & 2 & 2 &  &  &  &  &  & \bt 2716.40 &  &  &  &  &  \\
\bottomrule

	\end{tabular}
	\caption{
		We presents the effect of preprocessing in combination with different
		backend engines on model checking time.
		We compare no preprocessing to only phase abstraction without forwarding
		(PA) and full preprocessing (Full).
		Note that, ABC does not support reset functions and can therefore not be
		combined with full preprocessing.
		For each model we present the phase number without forwarding
		$\overline{n}$ for PA
		and the duration $d$ and phase number $n$ corresponding to Full.
		Models where the property holds are marked as safe.
		The first two rows present the number of solved instances and the PAR2
		score~\cite{froleyks2021sat} over all 818 benchmarks.
		The table shows all instances where preprocessing had either a positive or
		negative impact on model checking success,
		with the exception of those instances rendered unsolvable for our IC3
		implementation by forwarding.
	}
	\label{tab:mc}
\end{table}
\section{Experimental Evaluation}

\begin{figure}[t]
	\centering
	\includegraphics[width=0.6\textwidth]{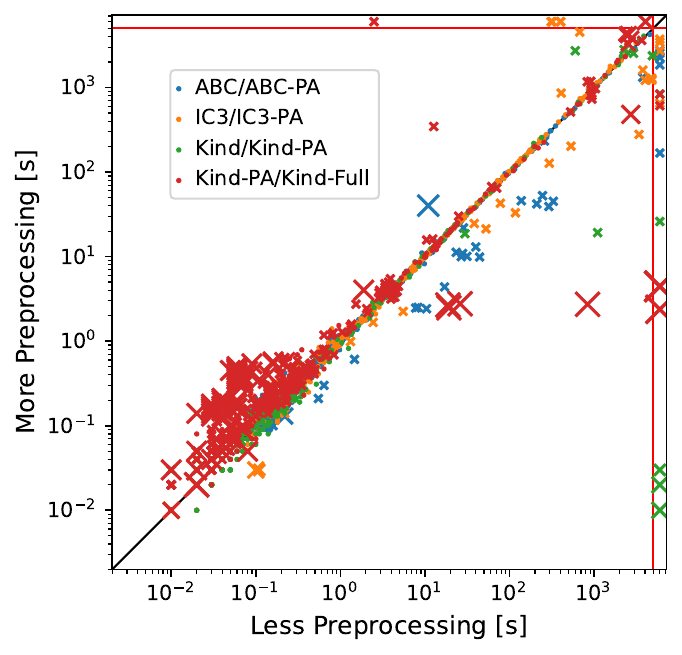}
	\caption{
		Comparison of model checking performance.
		We compare four pairs of configurations; the three backend engines with
		and without phase abstraction (with fixed duration 0)
		and for Kind we present the effect of additionally allowing
		forwarding.
		The size of the markers represents $n+d$.
		The dots represent instances where the preprocessing heuristic decided not
		to alter the circuit.
		The red lines mark the timeout of 5000 seconds.
		Markers beyond that line represent instances solved by one configuration but not the other.
	}
	\label{fig:mc}
\end{figure}

\begin{figure}[t]
	\centering
	\includegraphics[width=0.6\textwidth]{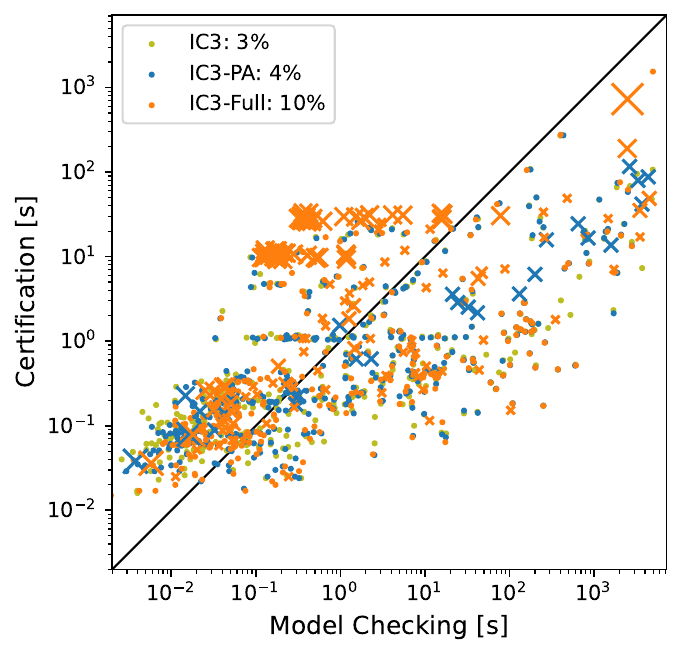}
	\caption{Certification vs. model checking time for three configurations of our
		IC3 engine.
		The legend shows the cumulative overhead of including certification for all solved instances.
		The size of the markers represents $n+d$.
		The dots represent instances where preprocessing did not
		alter the circuit.
	}
	\label{fig:cert}
\end{figure}
\label{sec:exp}
This section presents experimental results for evaluating the impact of
preprocessing on
the different backends, as well as the effectiveness of our proposed certification
approach. The experiments were run in parallel on a cluster
 of 32
 nodes. Each node was equipped
 with two 8-core Intel Xeon E5-2620 v4 CPUs running at 2.10 GHz and
 128 GB of main memory.
We allocated 8 instances to each node, with a
timeout of 5000 seconds for model checking and 10\,000 seconds for certificate checking.
Memory is limit to 8 GB per instance in both cases.

The benchmarks are obtained from
HWMCC2010~\cite{hwmcc10} which contains a good number of
industrial problems. As we observe from the experiments in
general, prepossessing is usually fast. Ignoring one outlier in our benchmark
set, it completes within an average of $0.07$ seconds and evaluates no more
than $17$ unfolding candidates per benchmark.
Interestingly, for the outlier ``bobsmnut1'', $3019$ unfoldings are
computed for $179$ different cube lassos within $34$ seconds.

Table~\ref{tab:mc} presents the effect of our preprocessing
on different backends, further illustrated in Fig~\ref{fig:mc}.
Our preprocessor was able to improve the performance of the sophisticated
IC3/PDR
implementation
in ABC, allowing us to solve five more instances, all from the intel family.
For each benchmark from this family, our heuristic computed an optimal phase
number of 2. A likely explanation for this is that the real-world industrial
designs tend to contain strict two-phase clocks~\cite{baumgartner1999model}.
The positive effect of phase abstraction is also clear in combination with the
$k$-induction (Kind) backend.
Circuit forwarding provides a further improvement, that is especially
notable on the prodcell benchmarks.
These also illustrate how forwarding enables more successful unfolding.
Without forwarding, preprocessing only unfolds 61 out of the 818 benchmarks with an average phase number of $2$,
with forwarding 152 circuits are unfolded with an average of $4$.

Even though our prototype implementation of IC3 is missing a number of
important features present in ABC, it still solves a large number of benchmarks.
However, as opposed to ABC it does lose a number of benchmarks with phase abstraction.
This can be explained by the lack of sophisticated rewriting that can exploit the unfolded circuits structure.
The addition of forwarding is highly detrimental to performance, losing 115 instances.
This is due to our implementation following the PDR design outlined in \cite{een2011efficient}.
It requires any blocked cube not to intersect the initial states after generalization.
If only a single reset state exists this check is linear in the size of the cube.
However, in the presence of reset functions it is implemented with a SAT call.
While also slower the main problem however is that the reset-intersection check is also more likely
to block generalization.
On the 115 lost benchmarks generalization failed 96\% of the time,
while it only failed in 1.8\% of the cases without forwarding.
We keep the optimization of our IC3 implementation in the presence of reset functions for future work.

Fig.~\ref{fig:cert} displays certification results on \mc in comparison to model
checking time. IC3 provides certificates that are easily verifiable,
as confirmed by our experiments with cumulative overhead of only 3\%.
The addition of phase abstraction (i.e., including constructing
$n$-folded witnesses as in Fig.~\ref{fig:pacert}, without witness back-warding)
does not bring significant additional overhead.
When forwarding is allowed, the certification overhead increases to 10\%.
The run time of certificates generation and encoding to SAT is negligible for all configurations.
The certification time is dominated by the SAT solving time for the transition (Def.~\ref{def:ressim}.2)
and consecution check.
Overall,
this is a significant improvement over related work from~\cite{fmcad23}
which reported
1154\%  overhead on the same
set of benchmarks using a $k$-induction engine as the backend.

\section{Conclusion}
In this paper, we present a certificate format that can be effectively validated
by
an independent certificate checker. We demonstrate its versatility by applying
it to an extended version of phase abstraction, which we introduce as one of the
contributions of this paper. We have implemented the proposed
approach
on a new certifying model checker \mc. The experimental results on HWMCC
instances
show
that our approach is effective and yields very small certification overhead, as a
vast improvement over related work.
Our certificate format allows for smaller certificates and is designed to be possibly used
in hardware model checking competitions as a standardized format.

Beyond increasing trust in model checking, certificates can be utilized in many other scenarios. For instance, 
such certificates will allow the use of model checkers as additional {hammers} in interactive theorem provers such as Isabelle~\cite{nipkow2002isabelle} via Sledgehammer~\cite{paulson2023sledgehammer}, with the potential of significantly reducing the effort
needed for using theorem provers in domains where model checking is essential, such as formal hardware verification,
our main application of interest.
Currently in Isabelle, Sledgehammer allows to encode the current goal for automatic theorem provers or SMT solvers and then call one of many tools to solve the problem.
The tool then provides a certificate which is lifted to a proof that can be replayed in Isabelle. We plan to add our model checker as an additional hammer to increase the automatic proof capability of Isabelle.
This further motivates us to investigate
certificate trimming via SAT proofs.

\FloatBarrier{}

	\subsubsection*{Acknowledgements}
	This work is supported by the Austrian Science Fund (FWF) under the 
       project W1255-N23, the LIT AI Lab funded by the State of Upper
	Austria, the ERC-2020-AdG
	101020093, the Academy of Finland under the project 336092 and by a gift from Intel Corporation.

\newpage
\bibliographystyle{splncs04}
\bibliography{paper}

\newpage
\section{Appendix}
This Appendix first provides two observations relating witness circuits under the new restricted simulation relation to the stratified simulation relation introduced in previous work~\cite{biere2022stratified}.
The next section provides formal proofs for the correctness of Theorem~\ref{thm:safety} and thereby the formal basis of our certification approach based on restricted simulation.
The final section formally proves the completeness of the witness circuit construction for phase abstraction.

\subsection{Comparison: Stratified and Restricted Simulation}
Previous work~\cite{biere2022stratified} introduced a more restricted form of simulation relation.
The following two propositions show that the new format subsumes this work,
and all previously introduced techniques for witness circuit construction can still be applied.

\begin{proposition}~\label{prop:strat}
	Given two stratified circuits $C=(I,L,R,F,P)$ and $C'=(I',L',R',F',P')$. If $C'$
	simulates $C$ under the
	stratified simulation relation~\cite{fmcad23}, then $C'$ also simulates $C$ under the
	restricted
	simulation.
\end{proposition}

\begin{proposition}~\label{prop:kwit}
	Given a circuit $C=(I,L,R,F,P)$ and its stratified $k$-witness circuit
	$C'=(I',L',R',F',P')$ as defined in~\cite{biere2022stratified}.
	Then $C'$ is a witness
	circuit of $C$ according to Definition \ref{def:wit}.
\end{proposition}

\subsection{Correctness: Restricted-Simulation-Based Witnesses}
The main claim for the correctness of our certification approach is given in Theorem~\ref{thm:safety}.
The formal proof of this theorem is split over the following lemmas relating to the reset, transition and property aspects of Def.~\ref{def:ressim}.
The first of these lemmas is presented in the main paper, and restated here for completeness.

\begin{replemma}
	Let $C=(I,L,R,F,P)$ and $C'=(I',L',R',F',P')$
	be two stratified circuits
	satisfying the reset condition defined in
	Def.~\ref{def:ressim}.\ref{def:ressim:reset}. Then
	$R'(L\cap L')$ is semantically dependent only on their common variables.
\end{replemma}

\begin{proof}
	We provide a proof by contradiction. Given a latch $l_a\in (L'\setminus L)$
	and a latch $l_b\in (L\cap L')$,
	suppose that  $r'_{l_b}(I',L')|_{l_a} \not\equiv r'_{l_b}(I',L')|_{\neg l_a}$
	which entails that $r'_{l_b}(I',L')$ is dependent on $l_a$.
	Since $l_a\in (L'\setminus L)$,  we have
	$l_a\notin L$. Based on this, we have $r_{l_b}(I,L)|_{l_a}\equiv
	r_{l_b}(I,L)|_{\neg l_a}$ since $l_a$ is not a variable in $C$. Let $s$ be a
	satisfying assignment to $r_{l_b}(I,L)|_{l_a}$ as well as $r_{l_b}(I,L)|_{\neg
		l_a}$. By the reset condition $r_{l_b}(I,L)\equiv r'_{l_b}(I',L')$, the same
	assignment $s$ satisfies $r'_{l_b}(I',L')|_{l_a}$ and $r'_{l_b}(I',L')|_{\neg
		l_a}$. Then we have $r'_{l_b}(I',L')|_{l_a}\equiv r'_{l_b}(I',L')|_{\neg l_a}$,
	which
	contradicts the fact that $r'_{l_b}(I',L')|_{l_a}\not\equiv
	r'_{l_b}(I',L')|_{\neg l_a}$.
	The same argument can be applied to the inputs.
	Therefore, every $r'_l(I',L')$ for $l\in
	(L\cap L')$ is dependent only on common variables $(L\cap L')\cup(I\cap I')$. Since
	$R'(L\cap L')$ is defined as $\bigwedge\limits_{l\in (L\cap L')} l\simeq
	r'_l(I',L')$,
	the same follows.
	\qed
\end{proof}

\begin{lemma}\label{lem:resext}
	Given two stratified circuits $C=(I,L,R,F,P)$ and $C'=(I',L',R',\\F',P')$ that
	satisfy Def.~\ref{def:ressim}.\ref{def:ressim:reset}.
	For every reset state in $C$, there is a
	reset state in $C'$ with the same assignments to $L\cap L'$.
\end{lemma}
\begin{proof}
	Let $s$ be a satisfying assignment to $R(L)$. Under the reset condition
	$r_l(I,L)\equiv r'_l(I',L')$ for $l\in (L\cap L')$,  the same
	assignment satisfies $R'(L\cap L')$. Since the reset functions of
	$L\cap L'$ can be assumed to not have
	dependencies on other latches $L'\setminus L$ according to
	Lemma~\ref{lem:konly}, and
	the rest of $R'(L')$ (i.e., $R'(L'\setminus L)$) is also stratified, we
	can assume a
	topological ordering of the dependency graph where all latches in
	$L'\setminus
	L$ come before $L\cap L'$. We can therefore assign values to the
	rest of the latches $L'\setminus L$ by traversing in reverse order.
	\qed
\end{proof}

That is a reset state in $C$ can always be extended to a reset state in
$C'$ and similar arguments can be applied to transitions from current to next state:

\begin{lemma}\label{lem:konly_f}
	Given two stratified circuits $C=(I,L,R,F,P)$ and $C'=(I',L',R',\\F',P')$ that
	satisfy the transition condition
	(Def.~\ref{def:ressim}).
	Then $f'_l(I',L')$ for all $l\in (L\cap L')$ semantically only depend
	on latches in $L\cap L'$.
\end{lemma}
\begin{proof}
	The proof  follows the same logic as that of
	Lemma~\ref{lem:konly}.\qed
\end{proof}
\begin{lemma}\label{lem:fext}
	Given two stratified circuits $C=(I,L,R,F,P)$ and $C'=(I',L',R',\\F',P')$ that
	satisfy the transition condition
	(Def.~\ref{def:ressim}.\ref{def:ressim:trans}). Let $u$ be a state in $C$ and
	$u'$ a state in $C'$ such that the common variables are assigned to the
	same values in both $u$ and $u'$.  Then for every successor of $u$, there is a
	successor state of $u'$ with
	the same
	assignments to $L\cap L'$.
\end{lemma}

\begin{proof}
	First of all, let $s$ be a satisfying assignment to $U_1$ (i.e., $L_1\simeq
	F(I_0,L_0)$). We construct a satisfying assignment $s'$ to
	$U'_1$ by first
	applying the same variable assignment to $L_0\cap L'_0$ in $C'$ (as well as
	$I_0\cap I'_0$) and extend
	it
	to an
	assignment of $I'_0 \cup L_0'$. Since $l_1\simeq f_{l}(I_0,L_0)$ for all
	$l\in L_0$, and the transition condition holds, by
	Def.~\ref{def:ressim:trans},
	$l_1\simeq f'_{l}(I'_0,L'_0)$ for all $l\in (L\cap L')$. Therefore, together
	with Lemma~\ref{lem:konly_f}, we can
	keep the same variable assignment to $L_1\cap L'_1$ in $s'$. The rest of the
	variables
	$L'_1\setminus L_1$ have their transition functions solely depend on
	variables from $I'_0$ and $L'_0$ therefore $s'$ is guaranteed to be
	satisfiable.


	\qed
\end{proof}

\begin{reptheorem}
	Let $C=(I,L,R,F,P)$ and
	$C'=(I',L',R',F',P')$ be two stratified circuits, where $C'$ simulates $C$
	under restricted simulation. \\If
	$C'$ is
	safe, then $C$ is also safe.
\end{reptheorem}
\begin{proof}
	We assume $C'$ is safe, and provide a proof by contradiction by assuming
	$C$ is not
	safe.
	Suppose there is an assignment $s$ over $I \cup L$ satisfying
	$R(L_0)\wedge U_m\wedge
		\neg P(I_m,L_m)$. By Lemma~\ref{lem:resext}, the same
	assignment of $R(L_0\cap L'_0)$ can be extended to satisfy $R'(L'_0)$.
	Furthermore, by Lemma~\ref{lem:fext}, we can construct a satisfying
	assignment for $R'(L'_0)\wedge U'_m$ such that the common latches are
	always assigned the same values as according to $s$.
	Thus, by assumption that $C'$ is safe, $P'(I'_m,L'_m)$ follows which together
	with
	Def.~\ref{def:ressim}.\ref{def:ressim:prop} results in
	a contradiction.\qed
\end{proof}

\subsection{Completeness: Phase Abstraction Witness Construction}
The presented witness circuit construction for phase abstraction will always produce a valid witness circuit.
This claim is presented for each step of the construction in Theorem~\ref{thm:reduced},~\ref{thm:compwit}, and~\ref{thm:nfolded}.
The proofs are presented in the following. For Theorem~\ref{thm:compwit}, the definition of unfolded loop invariant is restated in greater detail.

\begin{reptheorem}
	Given a circuit $C=(I,L,R,F,P)$ and its reduced circuit
	$C'=(I',L',R',F',P')$. A witness circuit of $C'$ is
	also a
	witness circuit of $C$.
\end{reptheorem}

\begin{proof}
	Let $W'=(J',M',S',G',Q')$ be a witness circuit of $C'$. First of all, we show
	that
	$W'$ simulates
	$C.$ As $L'\subseteq L, K =L'\cap M'$ is also the common set of
	latches between
	$C$ and $W'$. The resets of $C$ remain stratified. As the reset functions and
	transition functions of $K$ stay the
	same, together with Def.~\ref{def:reduced} where $P'=P$, the three checks in
	Def.~\ref{def:ressim} are satisfied. We conclude
	$W'$ simulates $C.$ By Def.~\ref{def:wit}, $Q'$ is an inductive invariant.
	Thus $W'$ is also a witness circuit of $C$.
\end{proof}

\begin{definition}[Unfolded cube lasso]
	Given a circuit $C=(I, L, R, F, P)$ with a phase number
	$n\in\mathbb{N^+}$, and its unfolded circuit $C'=(I',L',R',F',P')$.
	Let	$c_0,\cdots, c_d,\cdots,c_\delta,\cdots,c_{\delta+\omega}$ be the cube lasso
	of $C$. The unfolded cube lasso
	$c'_0,\cdots,c'_{\delta'},\cdots,c'_{\delta'+\omega'}$ is
	defined as follows: (i) $\delta'*n+d=\delta; \omega'*n + n-1=\omega.$
	(ii) For $i\in[0,\delta'+\omega'),
	c_i'=\bigwedge\limits_{j\in[0,n)}c_{i*n+j+d}(I^j_i,L^j_i).$

\end{definition}
\begin{definition}[Unfolded loop invariant]
	Given a circuit $C=(I, L,R, F, P)$ and its unfolded circuit
	$C'=(I', L',R', F',
	P')$, and
	a cube lasso $c_0,\cdots,c_{m}$ of $C$. Let $c'_0,{\ldots},c'_{m'}$
	be the unfolded cube lasso. \\
	The unfolded loop invariant $\phi$ is defined as
	$\bigvee\limits_{i\in[0,m']} c'_i.$
\end{definition}

\begin{lemma}
	The unfolded cube lasso is a cube lasso in the unfolded circuit.
\end{lemma}
\begin{lemma}
	Given a circuit $C=(I, L,R, F, P)$ and its unfolded circuit $C'=(I', L,R', F',
	P')$ with a phase number $n$. Let $\Gamma \subseteq L$ be a set of latches
	that are associated with periodic
	signals determined from a cube lasso
	$c_0,{\ldots},c_{\delta+\omega}$ of $C$. The unfolded loop invariant $\phi$ is an
	inductive
	invariant in $C'$ for the
	property
	$\bigwedge\limits_{l\in\Gamma}(\bigwedge\limits_{i\in[0,n]}(l^i\simeq
	\lambda_l^i))$ , where $l^i$ is a temporal
	copy of $l$.
\end{lemma}

In the above lemma, the property states in essence that the periodic signals
always have the values according to their periodic patterns in the unfolded
cube lasso. Each cube in the unfolded cube lasso is a
partial
assignment to $L'$, which consists of $n$ copies of the original set of latches $L$.
Therefore we
use temporal copies such as $\lambda_{l}^i$ to represent the periodic
values of
the
latches.

\begin{reptheorem}
	Given circuit $C=(I,L,R,F,P)$, and factor circuit $C'=(I',L',\\R',F',P')$. Let
	$W'=(J',M',S',G',Q')$ be a witness circuit of $C'$, and $W=(J,M,S, G,Q)$
	constructed as in Def.~\ref{def:compwit}. Then $W$ is a witness
	circuit of $C$.
\end{reptheorem}
\begin{proof}
	First of all we show the composite witness $W$ simulates $C$. $W'$ and $C$
	are
	stratified and $R$ only references
	latches in $L_1$ thus no new cyclic dependencies is introduced.
	Therefore $W$ is stratified too. $L$ is the common set of
	latches for $W$ and $C$. By Def.~\ref{def:fac}, the reset and transition
	functions of
	$L$ are the same in both circuits, this satisfies the reset check as well as the
	transition check. Since $W'$ is a witness circuit of $C'$, we have
	$Q'\Rightarrow P'$
	and therefore $Q\Rightarrow P'$. We omit the rest of the proof as it follows
	the same
	logic as
	Theorem 2 in \cite{fmcad23}. \qed
\end{proof}

\begin{reptheorem}
	Given a circuit $C=(I, L, R, F, P)$ with a phase number $n\in\mathbb{N^+}$,
	its
	unfolded cicuit
	$C'=(I',L',R',F',P')$ with a witness circuit $W'=(J',M', S', G', Q')$.
	Let $W=(J,
		M, S, G, Q)$ be the circuit constructed as in
	Def.~\ref{def:unfoldwit}. Then $W$ is a witness circuit of $C$.
\end{reptheorem}
\begin{proof}
	First, we prove the simulation relation. The $L^0$ latches have the same
	reset
	functions as in $C$, thus they are stratified, and they do not
	have dependencies on other latches outside $L^0$. The resets of $N$ are the
	same as in $W'$ thus also stratified. Based on Def~\ref{def:unfoldwit}, the
	rest
	of the latches do not depend on other latches therefore $S$ is stratified.
	Let
	$K=L\cap M$. By
	Def.~\ref{def:unfoldwit}, for $l\in L^0$, $s_l\equiv r'_l$.
	Together with Def.~\ref{def:unfold}, we have
	$R(K)\Rightarrow S(K)$. By Def.~\ref{def:unfoldwit}, $f_l\equiv
		g_l$ for $l\in K$.
	Also by $q^0$, we have $Q\Rightarrow P$. Therefore, $W$ simulates $C$.

	Next we show that the BMC check passes for $W$ such that
	$S(M)\Rightarrow Q(J,M)$. Since only $b^0$
	is set at reset, $q^1,q^2, q^3$ and $q^5$ are satisfied. All $e^i$s are set to
	$\bot$ at reset, thus $q^6,q^7,q^8$ are also satisfied.
	The reset in
	Def.~\ref{def:unfoldwit} directly implies $R'(L^0)$ and
	$S'(N)$, and by Def.~\ref{def:unfold} it also satisfies $R(L^0)$,
	which satisfies $q^4$. Based on Def.~\ref{def:ressim},
	$R'(L^0)\Rightarrow S'(M'\cap L^0)$, which together with the
	stratification of $S'$, results in $S'(M')$. This implies
	$Q'(J',M')$, based on the BMC check of the inductive invariant $Q'$.
	By Def.~\ref{def:ressim}, we have $P'(I',L')$, and based on
	Def.~\ref{def:unfold} this gives us $P(I^0,L^0)$ thus $q^0$ is also
	satisfied.

	Let
	$V_1$ be the unrolling of $W$ and
	we move on to prove $V_1\wedge
		Q(J_0,M_0)\Rightarrow Q(J_1,M_1)$ by providing a proof by contradiction.
	We assume
	$V_1\wedge
		Q(J_0,M_0)\wedge\neg Q(J_1,M_1)$ has a satisfying assignment, and
	fix this
	assignment. We consider two cases
	based on the values of
	$b^i_0$: (i) all $b^i_0$ are set to $\top$; (ii) $B_0$ is partially
	initialised, i.e., not
	all $b_0^i$ set to $\top$.
	We begin with the first case. Since $b^i$s always transition to
	$\top$
	or
	$b^{i-1}$, and under the assumption that all $b^i$s are $\top$,
	$q^1_1,q^2_1,q^4_1$ and $q^7_1$ are immediately satisfied.
	By Def.~\ref{def:unfoldwit}, $L^0$ transitions in the same way as in
	$C$, and the rest of the latches are simply
	copying the values during the transition, thus $q^3_1$ is satisfied.
	We move on to consider $q^5_1$.
	Based on
	$q^5_0$, we
	get a
	satisfying assignment for $\bigvee\limits_{i\in[0,n)}
		((\bigwedge\limits_{j\in[i,n-1)}\neg e^j)\wedge
		(\bigwedge\limits_{j\in[0,i)} e^j)\wedge
		Q'(J^0,M'\cap(L^i\cdots\cup
		L^{i+n-1}\cup N)) $. We thus consider three different cases
	based on the
	disjunction and the value of $e^i_0$s.
	\begin{itemize}
		\item Case where all $e^i_0$s are $\bot$: based on
		      Def.~\ref{def:unfoldwit},
		      $e^0_1$ is set to $\top$ and the rest set to $\bot$, thus $q^6_1, q^8_1$
		      are satisfied.
		      Since $e^{n-2}_0$ is set to $\bot$, the latches of $N$
		      stay the same after the transition. In this case, we already have
		      $Q'(J^0_0,M'_0\cap(L^0_0\cdots\cup
			      L^{n-1}_0\cup N_0)) $ satisfied, together with the transition function
		      defined and $q^3_0$, the same assignment satisfies
		      $(\bigwedge\limits_{j\in[1,n-1)}\neg e^j_1)\wedge e^0_1\wedge
			      Q'(J^0_1,M'_1\cap (L^1_1\cup
			      L^n_1\cup N_1))$, which satisfies $q^5_1$. By
		      Def.~\ref{def:ressim}, we have
		      $P'(I'_1,L'_1)$ which also implies $P(I^0_1,L^0_1)$. Therefore
		      $q_1^0$ is also satisfied.
		\item Case where not all $e^i_0$ set to $\top$: let $k$ be the index such
		      that
		      $0< k<n-2$,
		      $e^i_0\simeq \top$ for $i\in[0,k]$ and $e^i_0\simeq \bot$ for
		      $i\in(k,n-2]$,
		      and after the transition $e^i_1\simeq \top$ for $i\in[0,k+1]$ and
		      $e^i_1\simeq \bot$ for $i\in[k+2,n-2]$,
		      which immediately satisfies $q^6_1$ and $q^8_1$. Based on $q^5_0$,
		      we already have
		      $Q'(J^0_0,M'_0\cap (L^{k+1}_0\cdots\cup
			      L^{k+n}_0\cup N_0)) $. Similarly as the case above, $e^{n-2}_0$ is $\bot$
		      thus the rest follows.

		\item Case where all $e_0^i$s are $\top$: based on
		      Def.~\ref{def:unfoldwit},
		      all $e^i_1$s become $\bot$, thus $q^6_1$ and $q^8_1$ satisfied.  As
		      stated in
		      Def.~\ref{def:unfoldwit}, when $e_0\equiv \top$, latches in $N$
		      transition as in $W'$. In this case, based on $q^5_0$, we already have
		      $Q'(J^1_0,M'_0\cap (L^{n-1}_0\cdots\cup
			      L^{2n-2}_0\cup N_0))$ satisfied, together with the transition function
		      defined and $q^3_0$, the same assignment satisfies
		      $Q'(J^0_1,M'_1\cap (L^0_1\cup
			      L^{n-1}_1\cup N_1))$, which satisfies $q^5_1$. By
		      Def.~\ref{def:ressim}, we have
		      $P'(I'_1,L'_1)$ which implies $P(I^0_1,L^0_1)$. Therefore
		      $p_1^0$ is satisfied.

	\end{itemize}

	In either
	case, we can apply the same assignment to satisfy $Q(J_1,M_1)$ and
	therefore reach a contradiction. The rest of the proof follows similar
	logic.
	\qed
\end{proof}

\end{document}